\documentclass[journal]{IEEEtran}

\usepackage{amsmath}
\usepackage{amssymb}
\usepackage{amsthm}

\newtheorem{theorem}{Theorem}
\newtheorem{lemma}{Lemma}
\newtheorem{definition}{Definition}
\newtheorem{proposition}{Proposition}
\newtheorem{corollary}{Corollary}
\newtheorem{remark}{Remark}

\usepackage{tikz}
\usepackage{stackengine}

\usepackage{amsfonts}
\usepackage{bbm}
\usepackage{mathtools}
\usepackage{physics}
\usepackage[ruled,vlined,linesnumbered]{algorithm2e}
\usepackage{enumerate}
\usepackage{mathtools}
\usepackage{stackengine}
\usepackage{todonotes}

\usepackage{mathrsfs}
\usepackage{subcaption}

\usepackage{booktabs}
\usepackage{pgfplots}
\usepackage{array,colortbl,multirow,multicol,booktabs,ctable}

\usepackage{pifont}% http://ctan.org/pkg/pifont
\newcommand{\cmark}{\ding{51}}%
\newcommand{\xmark}{\ding{55}}%

\begin{document}
\bstctlcite{bibparams}

\title{Stochastic Channel Models \\ for Satellite Mega-Constellations}

\author{Brendon~McBain,~\IEEEmembership{Member,~IEEE,}
        Yi~Hong,~\IEEEmembership{Senior Member,~IEEE,}
        and~Emanuele~Viterbo,~\IEEEmembership{Fellow,~IEEE}% <-this % stops a space
\thanks{This research work was supported by the Australian Research Council (ARC) through the Discovery Project under Grant DP210100412.}%
\thanks{The authors are with the Department of Electrical and Computer
Systems Engineering, Monash University, Clayton, VIC 3800, Australia
(e-mail: brendon.mcbain@monash.edu, yi.hong@monash.edu, emanuele.viterbo@monash.edu).}%
}

\markboth{}%
{}

\maketitle

\begin{abstract}
A {general} satellite channel model is proposed for communications between a rapidly moving low Earth orbit (LEO) satellite in a mega-constellation and a stationary user on Earth. The {channel} uses a non-homogeneous binomial point process (NBPP) for modelling the satellite positions, marked with an { ascending/descending binary random variable for modelling the satellite directions}. Using the marked NBPP, we derive the probability distributions of {power} gain, propagation delay, and Doppler shift, resulting in a stochastic signal propagation model {for the mega-constellation geometry in isolation of other effects}. This forms the basis for our proposed channel model as a randomly time-varying channel. { The scattering function of this channel is {derived} to characterise how the received power is spread in the delay-Doppler domain.} Global channel parameters such as path loss and channel spread are analysed in terms of the scattering function. { The channel statistics and the global channel parameters closely match realistic orbit simulations of the Starlink constellation.}
\end{abstract}


\IEEEpeerreviewmaketitle

\section{Introduction}
\IEEEPARstart{I}n recent years, a growing number of satellites have been launched into {\em low Earth orbit (LEO)} to form mega-constellations aimed at expanding global Internet coverage \cite{AlHraishawi2023}. However, the high mobility of LEO satellites introduces unique challenges to communication systems \cite{Vatalaro1995}, such as high latency and high Doppler shifts, that are not yet fully understood. Compared to traditional mobile communications, the theoretical understanding of these satellite networks is still in its early stages. As the development of {\em sixth-generation (6G)} wireless networks approaches, there is a pressing need to explore optimal design practices for reliable communication in non-terrestrial satellite networks and their integration with terrestrial networks. A key step toward this goal is developing accurate channel models for satellite networks. This paper focuses on stochastic satellite channel models, with an emphasis on capturing the high mobility of satellites in mega-constellations.

% (1) LEO MODELS

For modelling the communications channel between a user and a LEO mega-constellation, we must model the satellite positions { on their orbit trajectories}. Since there are {hundreds or thousands of} rapidly moving satellites in a mega-constellation, stochastic models like the {\em point process (PP)} are appropriate for tractable descriptions of the satellite positions. In \cite{Okati2020a, Okati2020b}, Okati et al. proposed the {\em binomial point process (BPP)} to model $N$ total satellites as independent and uniformly distributed points on a satellite sphere. Since the BPP did not capture the varying density of satellites at different latitudes, a correction factor was included in the total number of satellites to give the effective number of satellites at different user latitudes. In \cite{Okati2021,Okati2022}, Okati et al. proposed the {\em non-homogeneous Poisson point process (NPPP)} to model the satellites using the probability density of satellites across latitudes, still assuming uniformity across longitudes. This model did not fix the total number of satellites, only the satellite density. In \cite{Okati2023}, the BPP model with the correction factor was extended to include discrete altitudes by distributing the satellites across multiple satellite spheres. {These models do not include satellite directions along their orbits, which is the extension that will addressed in this paper.}

% (2) CHANNEL MODELS

Modelling the satellite positions is sufficient for modelling the {power gain} due to moving satellites. However, it can only give a channel model of satellite communications after coherent detection, ignoring propagation delay and Doppler effects. This channel model was used to derive coverage probabilities and capacity bounds { for mega-constellations} \cite{Okati2020a,Okati2020b,Okati2021,Okati2022}. In \cite{Baeza2022}, a survey of thirteen channel models of {\em non-geostationary orbit (NGSO)} satellites was conducted, in which more than half of the channel models only considered link attenuation in various scenarios such as multipath or rain attenuation. { Link attenuation models are important and complement the mega-constellation models discussed earlier, since they must be combined for a complete channel model.} In this survey, it was remarked that almost none of the channel models considered Doppler effects. When Doppler effects between the user and the satellite are considered in the channel model, it is common to use the maximum Doppler shift to parameterise Jakes' Doppler spectrum \cite{ITU2019, 3GPP2020}. However, using this Doppler spectrum is a heuristic since it is not derived from the satellite mega-constellation geometry. { In contrast, the approach in this paper will characterise the Doppler shift based on a geometric model of mega-constellations. To highlight the novelty in this approach, the stochastic channel models surveyed in \cite{Baeza2022} are classified in Table \ref{table:channel_models} in terms of if the models are based on mega-constellation geometry or link attenuation, and if they include Doppler shift.}

\begin{table*}[]
    \centering
\begin{tabular}{lcccccl}\toprule
        Model type  & {Mega-constellation modelling} & Link modelling & Doppler shift inclusion \\\midrule
 BPP \cite{Okati2020a,Okati2020b,Okati2023} & \cmark & \cmark & \xmark \\
NPPP \cite{Okati2021,Okati2022} & \cmark & \cmark & \xmark \\
TDL \cite{3GPP2020}  & \xmark & \cmark & \cmark \\
Elevation-dependent \cite{Bischel1996,Enric2021} & \xmark & \cmark & \xmark \\
Small-scale fading \cite{ShadowedRician2003}  & \xmark & \cmark & \xmark\\
Rain fading \cite{Kanellopoulos2014}  & \xmark & \cmark & \xmark \\
Marked NBPP (proposed)  & \cmark & \cmark & \cmark \\
\bottomrule
\end{tabular}
\caption{A classification of the stochastic satellite channel models surveyed in \cite{Baeza2022}.}
\label{table:channel_models}
\end{table*}

% (3) DOPPLER BACKGROUND LITERATURE

In \cite{Ali1998}, the Doppler shift was characterised for a stationary user and a moving LEO satellite. This characterisation considered the time-dependent movement of the satellite on its orbit { under the assumption of a mega-constellation geometry with circular orbits}. On the other hand, stochastic Doppler characterisations using stochastic geometry have been studied in just a couple of scenarios. In \cite{Khan2020}, the Doppler distribution was derived for a single LEO satellite in a fixed position transmitting to a random user in a cell on Earth. In \cite{Seo2024}, the Doppler distribution was derived for a LEO satellite transmitting to an {\em unmanned aerial vehicle (UAV)} randomly moving in a swarm. { In \cite{AlHourani2024}, a closed-form distribution of a Doppler upper bound was derived for the scenario of a random LEO satellite from a mega-constellation transmitting to a stationary user. 
This scenario is the focus of this paper in the context of stochastic channel modelling, which will be shown to very accurately characterise the true Doppler shift distribution. In addition, the relationship between Doppler shift, propagation delay, and channel gain (inverse of free space path loss) is studied, which has not been considered in the literature for a stochastic model of the mega-constellation geometry.}

In satellite communications, many useful techniques exist to compensate for propagation delay and Doppler shift. %This raises a question: if these signal distortions can be corrected, why should they still be included in the channel model? 
{ Nonetheless, understanding the channel before compensation can aid the design and testing of compensation techniques. Moreover, there are many possible approaches to compensation and they are highly dependent on the configuration of the mega-constellation network, e.g., the particular inter-satellite handover strategy alone can greatly influence delay and Doppler shift~\cite{Voicu2024}. In addition, Doppler shift compensation is known to be challenging in cooperative satellite communications \cite{Caus2022,Caus2023}. This paper does not include compensation so that the scope of the results are not limited. However, the methods developed are sufficiently general to support the inclusion of an arbitrary compensation technique. For this reason, the channel model that will be proposed in this paper is ``general'' in that it can be simplified depending on the specific LEO communications system that is to be modelled.}

\subsection{Contributions}
{

The contributions in this paper are focused on providing a fully general description of stochastic models and developing the necessary theoretical tools to enable their application to any performance metric of interest. We further present a robust validation of these models, shedding light on both their strengths and weaknesses. In particular, we make the following contributions:

\begin{itemize}
    \item {\em Generalised stochastic mega-constellation model:} The BPP model in \cite{Okati2020a,Okati2020b} is extended to include a non-uniform latitude density using a {\em non-homogeneous BPP (NBPP)}, similarly to \cite{Okati2021,Okati2022} with the Poisson approximation. In addition, the NBPP is ``marked'' with a binary ascending/descending indicator variable so that the velocity vector of a random satellite position can be modelled. This novelty allows the point process to additionally model the Doppler shift due to satellite movement.

\item {\em Probability distributions of signal propagation properties:} The CDFs and PDFs of channel gain, delay, and Doppler are derived, where the latter solves an open problem from \cite{AlHourani2024}. This requires developing fundamental performance analysis tools, including the CDF of the user-satellite central angle (or elevation angle) $\sigma$ for a random satellite, denoted by $p_{\rm cap}(\sigma)$. The CDFs and PDFs of channel gain and delay are derived in terms of $p_{\rm cap}(\sigma)$, which is possible for any parameter that depends only on $\sigma$. Since this is not satisfied for Doppler shift, we derive a general integral for finding the CDF and PDF of parameters that depend on the satellite latitude and longitude.

\item {\em Stochastic channel model analysis}: The stochastic mega-constellation model and its probability distributions are used to form a satellite mega-constellation channel model. This channel model is completely described by the scattering function statistic, which is derived in terms of the delay-Doppler joint distribution and the average path loss. Since the scattering function is a fundamental property of stochastic channel models, it allows us to study the accuracy of the channel model jointly in terms of delay, Doppler, and channel gain, which had previously only been studied in terms of specific performance metrics based on channel gain.

\end{itemize}

Given the growing importance of stochastic modeling for performance analysis in satellite mega-constellation networks, these contributions address a key gap in the literature by providing valuable insights and analysis tools.

}

\subsection{Paper Outline} The organisation of the paper is as follows. In Section~II, the LEO orbits are modelled as circular orbits, which is the key assumption for the mega-constellation model in Section~IV. In Section~III, signal propagation properties (channel gain, propagation delay, and Doppler shift) are defined for a fixed satellite following the geometry of Section~II. These properties are derived in terms of spherical coordinates, ready for the statistical derivations in Section IV. In Section~IV, probability distributions of the signal properties are derived for a random satellite from the mega-constellation NBPP model. These distributions are used to analyse the stochastic channel model in Section~IV. In Section~V, the stochastic channel model is defined and analysed in terms of the derived scattering function, which depends on the distributions derived in Section~III.

\subsection{Preliminaries}

\subsubsection{Notation}

A random variable is denoted by a capital letter such as $X$. The probability that $X$ satisfies the event $\mathcal{A}$ is denoted by $\mathbb{P}(X\in\mathcal{A})$. The cumulative distribution function (CDF) of $X$ is $F_{X}(x) = \mathbb{P}(X \leq x)$, and the probability density function (PDF) of $X$ is $f_{X}(x) = \frac{d}{dx} F_X(x)$. The expectation of $g(X)$, for some function $g$, is $\mathbb{E}[g(X)] = \int_{\mathbb{R}} g(x) f_X(x) dx$. The Dirac delta distribution is denoted by $\delta(x)$ for $x\in\mathbb{R}$. A vector is denoted by a bold lower-case letter such as $\mathbf{u}$. The components of a three-dimensional vector $\mathbf{u}$ are denoted by $\mathbf{u} = (u_x, u_y, u_z)$. The unit vector of $\mathbf{u}$ is denoted by $\hat{\mathbf{u}} = \mathbf{u}/||\mathbf{u}||$. The indicator function $\mathbbm{1}\{x \in \mathcal{X}\}$ equals one if $x \in \mathcal{X}$ and zero otherwise. All numerical values are given to 3 significant figures.

\subsubsection{Coordinate systems}
The cartesian coordinate system in three dimensions is specified by the unit vectors $\hat{\mathbf{x}}$, $\hat{\mathbf{y}}$, and $\hat{\mathbf{z}}$.

The spherical coordinate system is used with coordinates $(R,\theta,\phi)$ according to the mathematics convention: $R$ is the radial distance, $\theta$ is the rotational angle (azimuthal angle) on the x-y plane, and $\phi$ is the polar angle (zenithal angle) relative to the z-axis. The {\em longitude} is $\theta$ in degrees, and the {\em latitude} is $\pi/2 - \phi$ in degrees. 

The {\em rotated} spherical coordinate system (or {\em local} coordinate system) is specified by the unit vectors
\begin{align}
    \begin{bmatrix}
           \hat{\mathbf{r}} \\
           \hat{\boldsymbol{\theta}} \\
           \hat{\boldsymbol{\phi}}
         \end{bmatrix} = \begin{bmatrix}
   \sin \phi \cos \theta & \sin \phi \sin \theta & \cos \phi  \\
    -\sin \theta & \cos \theta & 0\\
    \cos \phi \cos \theta & \cos \phi \sin \theta & - \sin \phi
\end{bmatrix} \begin{bmatrix}
           \hat{\mathbf{x}} \\
           \hat{\mathbf{y}} \\
           \hat{\mathbf{z}}
         \end{bmatrix}
\end{align}
relative to a reference point at spherical coordinates $(R, \theta, \phi)$. The radial unit vector $\hat{\mathbf{r}}$ is relative to the center of the sphere, and the angle unit vectors $\hat{\boldsymbol{\phi}}$ and $\hat{\boldsymbol{\theta}}$ are tangent to the latitude and longitude lines, respectively, forming a tangent plane at the reference point. %Thus, we refer to this as the local coordinate system.

\section{Satellite Constellation Geometry}

\begin{figure*}
    \centering
    \begin{subfigure}[b]{0.5\textwidth}
        \centering
        \includegraphics[trim={14.7cm 3cm 13cm 0},clip,width=0.95\textwidth]{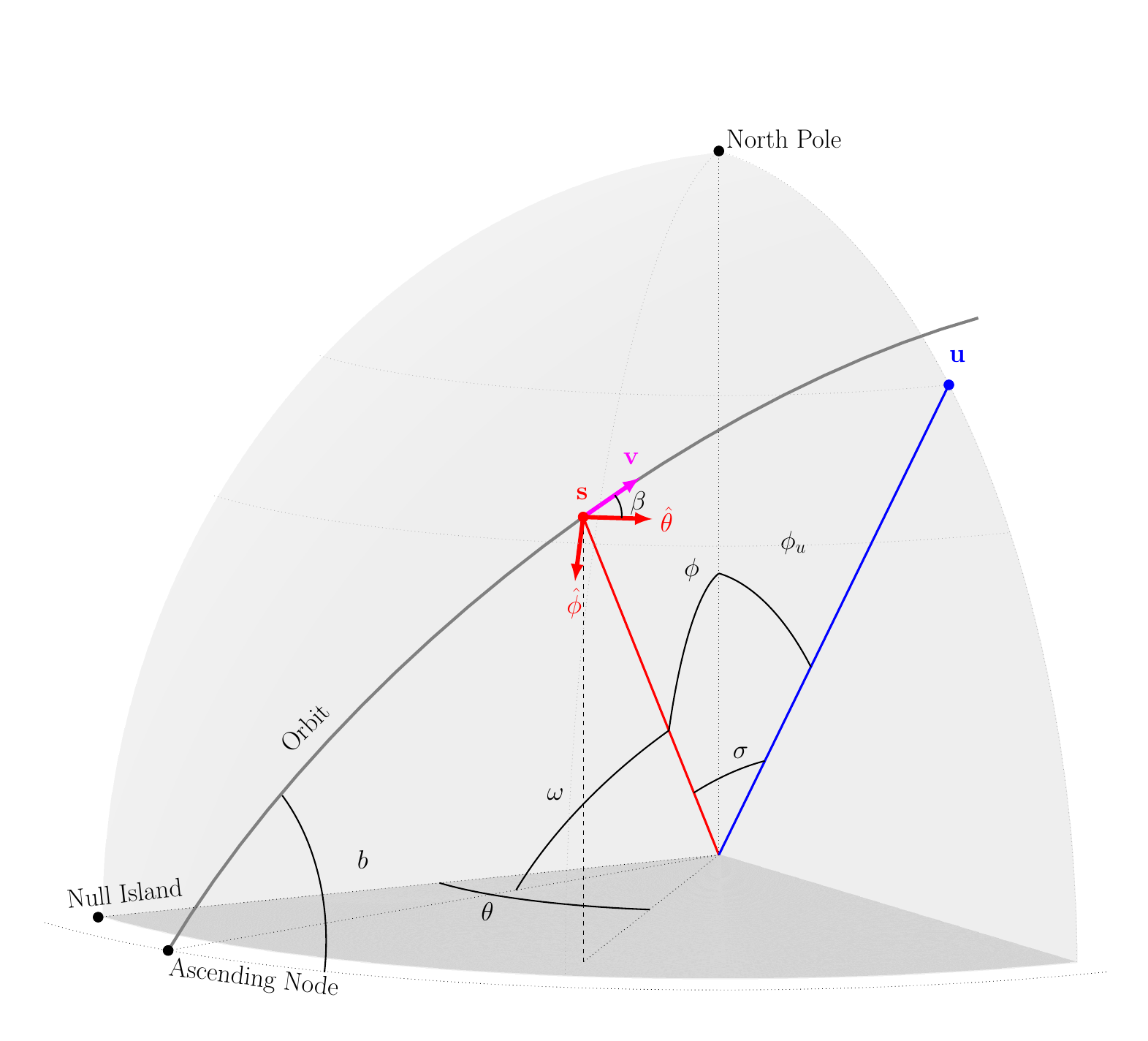}
        \caption{3D perspective view}
    \end{subfigure}%
    ~ 
    \begin{subfigure}[b]{0.5\textwidth}
        \centering
        \includegraphics[trim={1cm 0 1cm 0},clip,width=0.8\textwidth]{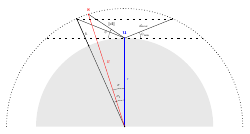}
        \caption{2D cross-sectional side view}
    \end{subfigure}
    \caption{The user and a serving satellite. (a) is a 3-dimensional view that captures the relative position of the user and the satellite orbit. (b) slices the Earth with a plane that intersects the user and the satellite.}
    \label{fig:user_sat}
\end{figure*}

\subsection{LEO Constellation Model}

Consider a {Walker-delta} mega-constellation for LEO satellites. For example, a single LEO satellite is illustrated in Fig. \ref{fig:user_sat}. All satellites are placed on an orbital shell that is a sphere of radius $R = r + h$, where $r$ is the radius of Earth and $h$ is the satellite altitude. We assume that each satellite is travelling at a constant speed $v$ m/s, tangentially to the Earth's surface. A satellite travels along its orbit in the easterly direction of increasing longitudes, referred to as a {\em prograde} orbit. The orbit of a satellite is specified by the position of its {\em ascending node} (the longitude on the equator that intersects the orbit) and the inclination angle $b$ (the angle between the equatorial plane and the orbital plane). The inclination angle restricts the magnitude of the latitude to be less than or equal to $b$. A mega-constellation is composed of multiple orbits with ascending nodes spaced at $s_{\rm{orb}}$ degrees. The position of a satellite on an orbit is specified by its {\em argument of latitude} $\omega$, which is the angle between the satellite and the ascending node on the inclined orbital plane. In the local coordinate system, the satellite travels with direction angle $\beta$ with respect to $\hat{\boldsymbol{\theta}}$ on the $\hat{\boldsymbol{\phi}}$-$\hat{\boldsymbol{\theta}}$ plane and varies with latitude. There are $N_{\rm{orb}}$ satellites on each orbit, with satellites spaced at $s_{\rm{sat}}$ degrees in terms of their argument of latitude, and there are $N = (360^{\circ}/s_{\rm{orb}}) N_{\rm{orb}}$ total satellites in the mega-constellation. 

The Walker-delta constellation can be simulated by choosing a time step $\Delta t$ that corresponds to increasing the argument of latitude in steps of $\Delta\omega = \omega_{\rm{sat}} \Delta t$, where $\omega_{\rm{sat}}=v/R$ is the satellite's angular velocity, which gives a sample period of $2\pi/\Delta\omega$. The satellite positions at each discrete time can be computed according to the circular orbit equations \cite{Lutz2000}, using the system parameters in Table \ref{tab:sys_params} for modelling the first and fourth orbital shells of Starlink. However, more realistic simulations can be performed using the {\em Simplified General Perturbations 4 (SGP4)} orbit model that considers perturbations (e.g., due to Earth's shape, drag, radiation, and gravitation effects from the sun and the moon) that cause the satellites to deviate from their average circular orbits. With this model, realistic orbits can be simulated using a {\em two-line element (TLE)} file of the constellation, which is publicly available for Starlink. From this TLE file, we simulate all satellites with parameters approximately equal to those in Table \ref{tab:sys_params} ($b \pm 1^\circ$ and $h \pm 50$ km). {The SGP4 simulations in this paper will use a starting epoch date of 2024/01/01 for $10$ intervals of $10$ minute durations ($100$ minutes in total).}

\begin{table}[]
    \begin{center}
    \begin{tabular}{cc}
\toprule
\textbf{Parameters} &  {\textbf{Values}}\\
\midrule

Satellite speed, $v$ & $7.29$ km/s\\
Carrier frequency, $f_c$ & $12.7$ GHz\\
Orbital inclination, $b$ & $53^\circ$\\
Altitude, $h$ & $550$ km\\
Num. of satellites, $N$  & $3168$\\
Num. of satellites per orbit, $N_{\rm{orb}}$ & $22$\\
Orbital spacing, $s_{\rm{orb}}$ &  $2.5^\circ$\\
\bottomrule
\end{tabular}
\end{center}
    \caption{System parameters based on the { first and fourth shells of the} Starlink mega-constellation \cite{Capez2023, Capez2024}. (The shells actually have orbital spacings $5^\circ$, with otherwise almost identical parameters, however since their ascending nodes are offset by $2.5^\circ$ we approximate them as one shell with a halved spacing).}
    \label{tab:sys_params}
\end{table}

\subsection{User-satellite Central Angle and Elevation Angle}

The angle between the user and the satellite on the plane that connects them is called the user-satellite {\em central angle} $\sigma$. This is a convenient parameter for expressing quantities that only depend on the user-satellite distance, forming a cone above the user with rotational symmetry. Without loss of generality, assume the user has a rotational angle of $\theta_{\rm{u}}=\pi/2$, putting the user on the y-z plane, then
\begin{align}
    \sigma = \cos^{-1}(\cos(\phi_{\rm{u}})\cos(\phi) + \sin(\phi_{\rm{u}})\sin(\phi)\sin(\theta))
\end{align}
which is often used for computing the {\em great-circle distance} between two points on a sphere. 

In addition, the angle between the user's horizon plane and the satellite is the {\em elevation angle $\psi$}, which is a parameter of the {\em topocentric coordinate system} that is commonly used since its parameters are directly observable by the user. The elevation angle $\psi$ relates to the central angle $\sigma$ through the user-satellite distance 
\begin{align}
    ||\mathbf{d}|| &= r\left[\sqrt{(R/r)^2 - \cos^2(\psi)}-\sin(\psi) \right] \\
    &= \sqrt{ r^2 +  R^2 - 2 r R \cos \sigma } .
\end{align}
In practice, the elevation angle must satisfy $\psi \geq \psi_{\min}$ for a minimum elevation angle $\psi_{\min}$ that depends on: (i) the visible region of the user with a line-of-sight to the satellite sphere, and (ii) the footprint size of the satellite beam \cite{Cakaj2014}. The maximum user-satellite distance $d_{\max}$ is achieved when a satellite satisfies $\psi = \psi_{\min}$. This corresponds to the {\em maximum central angle} $\sigma_1 = \sigma_1(\psi_{\min})$. Note that when $\psi_{\min} = 0^\circ$ (i.e., the user sees everything above their horizon plane) we have $\sigma_1 = \cos^{-1} (r/R)$ since $d_{\max} = R^2 - r^2$.

\subsection{Visible Cap of a User}
Assume that a user can only see satellites within a cone centred at the user with a central angle equal to $\sigma_1$, where satellites within the cone have central angles $\sigma = \sigma(\theta,\phi) \leq \sigma_1$. While a user is capable of observing a satellite within the cone, we must consider the fact that some orbits cannot be seen for users at high latitudes due to the orbital inclination $b$. This results in a void of satellite near the poles since satellites do not pass above latitude $b$ nor below latitude $-b$, which corresponds to polar angles $\overline{b} \leq \phi \leq \pi - \overline{b}$ with polar inclination angle $\overline{b} = \pi/2 - b$. This corresponds to a cap on the satellite sphere with radius determined by the central angle $\sigma_1$, sliced near the poles due to orbital inclination, giving the cap surface
\begin{align}\label{eq:vis_cap}
    \mathsf{Cap} &= \{(\theta, \phi) : 0 \leq \theta \leq 2\pi, \notag\\
    &\quad\quad\quad\quad\quad \overline{b} \leq \phi \leq \pi - \overline{b}, \notag\\
    &\quad\quad\quad\quad\quad \sigma(\theta,\phi) \leq \sigma_1\} .
\end{align}
Within this cap, a user will only see satellites with central angles in $[\sigma_{\min}, \sigma_{\max}]$. 

Consider a user in the northern hemisphere with $0 \leq \phi_{\rm{u}} \leq \pi/2$, since the southern hemisphere follows by symmetry. The central angle bounds $\sigma_{\min}$ and $\sigma_{\max}$ are characterised by considering two scenarios:
\begin{itemize}
    \item $\phi_{\rm{u}} \geq \overline{b}$: The closest satellite occurs at $\sigma = 0$ (directly overhead). The furthest occurs at $\sigma=\sigma_1$, since if the visible cap is sliced the other side will not be sliced. 
    \item $\phi_{\rm{u}} < \overline{b}$: The closest satellite occurs at central angle
    $\sigma=\bar{b}-\phi_{\rm{u}}$, due to the slicing at polar angle $\phi=\overline{b}$. The furthest occurs at $\sigma=\sigma_1$, since the slicing only occurs at lower polar angles. If the user is too far above the polar angle $\phi=\bar{b}$, they do not see satellites and the central angle is undefined.
\end{itemize}
These scenarios correspond to piecewise bounds
\begin{align}
\begin{cases} 
    \sigma_{\min} = 0, \, \sigma_{\max} = \sigma_1 & \text{if } \phi_{\rm{u}} \geq \overline{b}, \text{else}\\
    \sigma_{\min} = \overline{b} - \phi_{\rm{u}}, \, \sigma_{\max} = \sigma_1 & \text{if }   \overline{b} - \phi_{\rm{u}} \leq \sigma_1 , \\
    \text{undefined} & \text{if }  \overline{b} - \phi_{\rm{u}} > \sigma_1  .
\end{cases}
\end{align}
where the last two cases correspond to $\phi_{\rm{u}} < \overline{b}$. {Note that this is the support of the satellite central angle when the orbital spacing in a Walker-delta constellation is assumed to be arbitrary, otherwise it will be slightly different.}

\section{Signal Propagation Properties: Gain, Delay, and Doppler Shift}

%In the following, without loss of generality we may assume the user is at the north pole and the satellite is at polar angle $\phi$. Oftentimes, this is more convenient for proofs and easier to visualise. For an arbitrarily positioned user, we can replace $\phi$ with the user-satellite central angle $\sigma$.

\begin{proposition}\label{prop:pl}
    The channel gain of a satellite channel (inversely proportional to the free space path loss)  at user-satellite central angle $\sigma$ is 
\begin{align}
    \mathsf{G}(\sigma) &=  ||\mathbf{d}||^{-2} = (r^2 +  R^2 - 2 r R \cos\sigma)^{-1}
\end{align}
which is a monotonically decreasing function over domain $\sigma_{\min} \leq \sigma \leq \sigma_{\max}$. The inverse function of the channel gain $g=\mathsf{G}(\sigma)$ is
\begin{align}
    \mathsf{G}^{-1}(g) &= \cos^{-1}\left(\frac{g\left(r^{2}+R^{2}\right)-1}{2 g rR}\right) 
\end{align}
over domain $g_{\min} = \mathsf{G}(\sigma_{\max}) \leq g \leq g_{\max} = \mathsf{G}(\sigma_{\min})$.
\end{proposition}

\begin{proposition}\label{prop:delay}
The propagation delay of a satellite at user-satellite central angle $\sigma$ is 
\begin{align}
    \mathsf{T}(\sigma) &= \frac{ ||\mathbf{d}||}{c} = c^{-1}\sqrt{r^2 +  R^2 - 2 r R \cos \sigma}
\end{align}
where $c$ is the speed of light, which is a monotonically increasing function over domain $\sigma_{\min} \leq \sigma \leq \sigma_{\max}$. The inverse function of propagation delay $\tau = \mathsf{T}(\sigma)$ is
\begin{align}
    \mathsf{T}^{-1}(\tau) &= \cos^{-1}\left(\frac{r^2 + R^2 - c^2 \tau^2}{2 r R}\right) 
\end{align}
over domain $\tau_{\min} = \mathsf{T}(\sigma_{\min}) \leq \tau \leq \tau_{\max} = \mathsf{T}(\sigma_{\max})$.
\end{proposition}

These propositions are easily shown using monotonicity over the interval of allowable central angles $[\sigma_{\min}, \sigma_{\max}]$.

Numerous formulas of the Doppler shift can be found in the literature. It is often bounded in terms of the central angle (or elevation angle) \cite{3GPP2020}. {Under the circular orbits assumption, with a constant satellite speed,} the exact Doppler shift was derived in \cite{Ali1998} as a function of the elevation angle of the satellite and the maximum possible elevation angle of a satellite on the same orbit. This was demonstrated to accurately characterise LEOs. However, the results to follow in this paper require a Doppler shift function in terms of spherical coordinates, rather than topocentric coordinates.

\begin{theorem}\label{thm:doppler}
The (normalised) Doppler shift of a satellite at $(R,\theta,\phi)$ relative to a user at $(r,\pi/2,\phi_{\rm{u}})$ is 
\begin{align}\label{eq:doppler_shift}
    \mathsf{V}_a(\theta,\phi) &= \frac{vr}{||\mathbf{d}||}[ {{-}}\cos(\beta_a(\phi))\cos(\theta)\sin(\phi_{\rm{u}})\notag\\
    &\quad\quad\quad -\sin(\beta_a(\phi))(\sin(\phi)\cos(\phi_{\rm{u}}) \notag\\
    &\quad\quad\quad\quad\quad\quad\quad\quad - \cos(\phi)\sin(\theta)\sin(\phi_{\rm{u}})) ]
\end{align}
with direction angle
\begin{align}
    \beta_a(\phi) &= a\cos^{-1}\left(\frac{\cos(b)}{\sin(\phi)}\right)
\end{align}
that depends on whether the satellite is on an ascending orbit ($a=+1$) or a descending orbit ($a=-1$).
\end{theorem}
The proof is given in Appendix A.

For a carrier frequency $f_c$, the Doppler shift is $(f_c/c) \mathsf{V}_a(\theta,\phi)$ where $c$ is the speed of light. The maximum Doppler shift is $\nu_{\max} = \max\{\nu_{\max,a} : a \in \{\pm 1\}\}$ for $\nu_{\max,a} = \max\{\mathsf{V}_a(\theta,\phi) : (\theta,\phi) \in \mathsf{Cap} \}$. Observe that $\beta_a(\pi/2) = b$ and that $\beta_a(\phi)$ is flat around $\phi=\pi/2$, hence a tight approximation for a user at the equator is $\beta_a(\phi) \approx b$.

\section{Stochastic Signal Propagation Model}

In this section, we model the positions of $N$ satellites and their directions on the satellite sphere in a mega-constellation using a NBPP. This assumes that the satellites are independent and randomly distributed over the orbital shell. The distributions we consider are parameterised in terms of a random rotational angle $\Theta$, random polar angle $\Phi$, and random direction $\beta_A(\Phi)$.

\begin{definition}[Mega-constellation NBPP]
    The mega-constellation NBPP is defined by {$N$} i.i.d. satellites on a sphere with spherical coordinates $(R,\Theta,\Phi)$, where:
    \begin{itemize}
        \item Rotational angle $\Theta$ is uniformly distributed on $[0,2\pi]$.
        \item Polar angle $\Phi$ obeys the PDF \cite[Lemma~2]{Okati2020b} 
        \begin{align}
    f_{\Phi}(\phi) &= \frac{\sin(\phi)}{\pi \sqrt{\sin^2(\phi)-\cos^2(b)}} 
\end{align}
for $\phi \in [\overline{b},\pi-\overline{b}]$ and zero otherwise.
    \end{itemize}
In addition, it is marked by the direction angle $\beta_A(\Phi)$ where  $A$ is uniformly distributed on $\{\pm 1\}$.
\end{definition}

The rotational angle $\Theta$ is chosen to be uniformly random since the orbits are uniformly spaced by discrete rotational angles. The polar angle $\Phi$ is chosen to obey this particular PDF since it is the distribution that arises when choosing the argument of latitude $\omega$ uniformly random on $[0,2\pi]$ \cite[Lemma~2]{Okati2020a}. We note that the following results can be directly applied to an arbitrarily distributed $\Phi$, however, we require that $\Theta$ is uniformly distributed over $[0,2\pi]$ for rotational symmetry. In addition, the direction is uniformly distributed on $\{\pm 1\}$ since each satellite spends half of its time in each hemisphere.

\subsection{Average Fraction of Satellites}
In this section, we derive the { average fraction of satellites in the cap $\mathsf{Cap}(\sigma)$ with max. central angle $\sigma$, which is defined as $p_{\rm{cap}}(\sigma)=\mathbb{E}[\mathbbm{1}\{(\Theta,\Phi) \in \mathsf{Cap}(\sigma)\}]$}. However, we allow $\sigma_1$ to vary so that we can specify the probability of a satellite in an arbitrary visible cap. This will be useful in the next section for parameterising the cap probability in terms of other useful quantities.

\begin{theorem}\label{thm:p_cap}
    The average fraction of satellites in the visible cap centered at polar angle $\phi_{\rm{u}}$, $0 \leq \phi_{\rm{u}} \leq \pi/2$, with central angle $\sigma$, $0 \leq \sigma \leq \sigma_1$, is
    \begin{align}
        p_{\rm{cap}}(\sigma) &= F_{\Phi}(\max\{0, \sigma_1 - \phi_{\rm{u}}\}) \notag\\
        &\quad+ \frac{1}{2\pi}\int_{\max\{0, \sigma_1 - \phi_{\rm{u}}\}}^{\sigma_1 + \phi_{\rm{u}}} f_{\Phi}(\phi) L_1(\phi; \sigma) d\phi 
    \end{align}
    where 
    \begin{align}
         &L_1(\phi; \sigma)\notag\\
         &= \pi + 2\sin^{-1}(\csc(\phi)(\cot(\phi_{\rm{u}}) \cos(\phi) - \cos(\sigma)\csc(\phi_{\rm{u}}))) .
    \end{align}
\end{theorem}
The proof is given in Appendix B.

The success probability of a satellite in the NBPP is $p_{\rm{sat}} = p_{\rm{cap}}(\sigma_1)$. Then, the probability of $n$ visible satellites is the binomial distribution
\begin{align}
     p_{\rm vis}(n) &= {N \choose n} p^n_{\rm{sat}} (1-p_{\rm{sat}})^{N-n}
\end{align}
for $0 \leq n \leq N$. In particular, the probability of at least one satellite is defined as the {\em availability probability} $p_{\rm{a}} = 1 - p_{\rm{vis}}(0) = 1 - (1-p_{\rm{sat}})^{N}$. For a well-designed mega-constellation, we expect that $p_{\rm{a}}\approx 1$ so that there is always a satellite available for transmission to an 
arbitrary user on Earth. The average number of visible satellites for a user at polar angle $\phi_{\rm{u}}$ is specified by $Np_{\rm{sat}}$, allowing us to study the satellite coverage for a range of user latitudes. This is important for designing satellite networks with sufficient resources to equally serve users at any location  { and has been studied extensively in the literature (e.g., \cite[Fig. 2]{Hall2022}, \cite[Fig. 5(a)]{Duggan2022} and \cite[Fig. 4]{DelPortillo2019}), which we reiterate here with our tractable model.} { Note that multiple orbital shells would combine $p_{\rm sat}$ for different values of $b$ and $h$, as in \cite[Fig. 4]{DelPortillo2019} for the full Starlink constellation}. %This may additionally be scaled by the density of users at polar angle $\phi_{\rm{u}}$ to account for the varying density in users to be served.

Consider the average fraction of satellites in Fig. \ref{fig:prob_sat} for the system parameters in Table \ref{tab:sys_params}. A significant fraction of satellites are above users at a wide range of latitudes. Users at moderately-high latitudes just below $b = 53^\circ$ benefit from the higher density of satellites within sight. However, this density goes to zero above latitude $b$ and thus users above this latitude observe fewer and fewer satellites overhead, until eventually there are no satellites in sight. In addition, the average fraction of satellites as a function of user latitude is scaled up and slightly widened when reducing the minimum elevation angle $\psi_{\min}$. For $\psi_{\min} = 30^\circ$ and {$N = 3168$} satellites, a user at zero latitude (the equator) observes {$\approx 9.6$} satellites on average, which more than doubles for a user at latitude $b=53^\circ$ to {$\approx 25.6$} satellites on average. Up to user latitude $\approx 60^\circ$, there is at least one visible satellite $99\%$ of the time, otherwise there are none. Notably, above this latitude the density of the human population is close to zero.

The average fraction of satellites of the NBPP model tightly approximates the { circular orbits and SGP4} simulations. Observe that the model accuracy degrades at low latitudes, which will be explained in the next sub-section using more granular properties of the model.

\begin{figure}
    \centering
\includegraphics[width=0.95\linewidth]{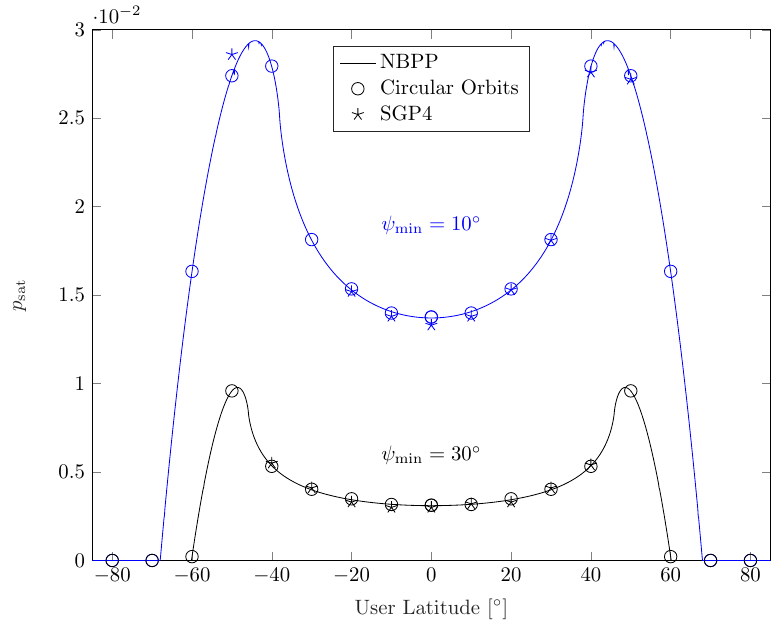}
    \caption{ Average fraction of satellites for the system parameters in Table \ref{tab:sys_params}. Two typical values of minimum elevation $\psi_{\min}$ are shown.}
    \label{fig:prob_sat}
\end{figure}

\subsection{Distribution of Gain, Delay, and Doppler Shift}

Let us now characterise the stochastic signal properties of satellites in the NBPP model: the gain is a random variable $H$, the propagation delay is a random variable $T$, and the Doppler shift is a random variable $V$. For convenience, we condition these random variables on the satellites being in the user's visible cone, which occurs with probability $p_{\rm{sat}}$, giving the visible gain $\overline{G}$, visible delay $\overline{T}$, and visible Doppler $\overline{V}$.

Firstly, we are going to derive the CDFs of these random variables by reparameterising the $p_{\rm{cap}}(\sigma)$ function, which is a CDF parameterised in terms of a central angle $\sigma$ of a cap. Since the signal properties are monotone functions, we can do a change of variables to find their CDFs. 

\begin{theorem}\label{thm:cdfs}
    Given a visible satellite, the CDF of its gain $\overline{G}$ is
    \begin{align}
        F_{\overline{G}}(g)  &= 1 - \frac{p_{\rm{cap}}(\mathsf{G}^{-1}(g))}{p_{\rm{sat}}} ,
    \end{align}
    and the CDF of its propagation delay $\overline{T}$ is
    \begin{align}
        F_{\overline{T}}(\tau)  &= \frac{p_{\rm{cap}}(\mathsf{T}^{-1}(\tau))}{p_{\rm{sat}}} .
    \end{align}
\end{theorem}
\begin{proof}
 Observe that $p_{\rm cap} = \mathbb{P}(\sigma(\Theta,\Phi) \geq \sigma_1)$, which is an inverse CDF. Now, observe that $\{\mathsf{T}(\sigma(\Theta,\Phi)) \leq \tau\} = \{\sigma(\Theta,\Phi) \geq \mathsf{T}^{-1}(\tau)\}$ and $\{\mathsf{G}(\sigma(\Theta,\Phi)) \leq g\} = \{\sigma(\Theta,\Phi) \leq \mathsf{G}^{-1}(g)\}$. This holds since delay and gain are one-to-one functions that uniquely specify $\sigma$. 
\end{proof}

We can find their PDFs by taking the derivative of the CDFs, which will be in terms of the derivative of $p_{\rm{cap}}(\sigma)$. These are simpler to derive using derivatives with respect to $\cos(\sigma)$, so that the $\cos^{-1}(\cdot)$ in the inverse functions is eliminated.

\begin{lemma}
    The derivative of $p_{\rm{cap}}(\sigma)$ with respect to $\cos(\sigma)$ is
    \begin{align}
    p'_{\rm{cap}}(\sigma) &=
    \frac{1}{2\pi}\int_{\max\{0, \sigma_1 - \phi_{\rm{u}}\}}^{\sigma_1 + \phi_{\rm{u}}} f_{\Phi}(\phi) L'(\phi; \sigma) d\phi
\end{align}
where
\begin{align}
    &L'(\phi; \sigma)\notag\\
    &= \frac{\partial}{\partial \cos(\sigma)} L_1(\phi; \sigma)\\
    &= - \frac{2 \csc(\phi_{\rm{u}}) \csc(\phi)}{\sqrt{1 - (\cot(\phi_{\rm{u}}) \cot(\phi)  - \cos(\sigma) \csc(\phi_{\rm{u}}) \csc(\phi))^2}} .
\end{align}
\end{lemma}
\begin{proof}
    The derivative exists since the integrand is continuous over the domain considered. %One can find $L'(\phi; \sigma)$ using Mathematica.
\end{proof}

\begin{theorem}\label{thm:pdfs}
     Given a visible satellite, the PDF of its gain $\overline{G}$ is
    \begin{align}
        f_{\overline{G}}(g)  &= -\frac{1}{2 g^2 r R} \left[ \frac{p'_{\rm cap}\left(\mathsf{G}^{-1}(g)\right)}{p_{\rm sat}} \right] ,
    \end{align}
    and the PDF of its propagation delay $\overline{T}$ is
    \begin{align}
    f_{\overline{T}}(\tau) &= -\frac{c^2 \tau}{rR} \left[ \frac{p'_{\rm cap}\left(\mathsf{T}^{-1}(\tau)\right)}{p_{\rm sat}} \right] .
\end{align}
\end{theorem}
\begin{proof}
    Apply the chain rule $\frac{d}{dy} \{p_{\rm{cap}}\left(f(y)\right)\} =  p'_{\rm{cap}}\left(f(y)\right) \frac{d}{dy}\left\{\cos(f(y))\right\}$.
\end{proof}

\begin{theorem}\label{thm:Doppler_CDF}
    Given a visible satellite, the CDF of its Doppler shift $\overline{V}$ conditioned on $A$ is
\begin{align}
    &F_{\overline{V}|A}(\nu|a) \notag\\
    &= \frac{1}{2\pi p_{\rm{sat}}}\int_{0}^{\pi} \int_{\theta_{\rm{u}} - L_1(\phi ; \sigma_1)/2}^{\theta_{\rm{u}} + L_1(\phi ; \sigma_1)/2} f_{\Phi}(\phi) \mathbbm{1}{\{\mathsf{V}_a(\theta,\phi) \leq \nu\}}  d\theta d\phi 
\end{align}
for all $\nu \in [-\nu_{\max},\nu_{\max}]$, $\tau \in [\tau_{\min},\tau_{\max}]$, and $a \in \{\pm 1\}$. Hence, we also have $F_{\overline{V}}(\nu) = 0.5 F_{\overline{V}|A}(\nu|{+1}) + 0.5 F_{\overline{V}|A}(\nu|{-1})$.
\end{theorem}

The proof is given in Appendix C.

%\begin{remark}
%    Observe that {$p_{\rm{cap}}(\sigma) = F_{\overline{V}|A}(\nu_{\max}|a)\big \vert_{\sigma_1=\sigma}$} and thus $p_{\rm{sat}} = F_{\overline{V}|A}(\nu_{\max}|a)$. However, we recommend using the single integral in Theorem \ref{thm:p_cap} rather than the double integral in Theorem \ref{thm:Doppler_CDF} for numerical computations.
%\end{remark}

Note that a closed-form PDF of the Doppler shift is challenging to compute in practice, while the CDF above is readily computed numerically. Instead, we approximate it by applying the finite difference method with resolution $\Delta\nu = 2.65$ kHz in the examples to follow.

A simple extension to the Doppler CDF is the delay-Doppler CDF. This will be useful in the next section where we study the joint relationship between delay and Doppler in the proposed channel model.

\begin{corollary}\label{thm:DD_CDF}
     Given a visible satellite, the CDF of its delay and Doppler $(\overline{V},\overline{T})$ conditioned on $A$ is
\begin{align}
    F_{\overline{V},\overline{T}|A}(\nu,\tau|a) &=  F_{\overline{V}|A}(\nu|a)\big \vert_{\sigma_1=\mathsf{T}^{-1}(\tau)}
\end{align}
for all $\nu \in [-\nu_{\max},\nu_{\max}]$, $\tau \in [\tau_{\min},\tau_{\max}]$, and $a \in \{\pm 1\}$. Hence, we also have $F_{\overline{V},\overline{T}|A}(\nu,\tau) = 0.5 F_{\overline{V},\overline{T}|A}(\nu,\tau|{+1}) + 0.5 F_{\overline{V},\overline{T}|A}(\nu,\tau|{-1})$.
\end{corollary}
\begin{proof}
    The delay function $\mathsf{T}$ is one-to-one over its domain and thus a given delay $\tau$ uniquely specifies the user-satellite central angle as $\sigma = \mathsf{T}^{-1}(\tau)$.
\end{proof}

\begin{figure*}
    \centering
    \begin{subfigure}[b]{0.5\textwidth}
        \centering
        \includegraphics[width=1\textwidth]{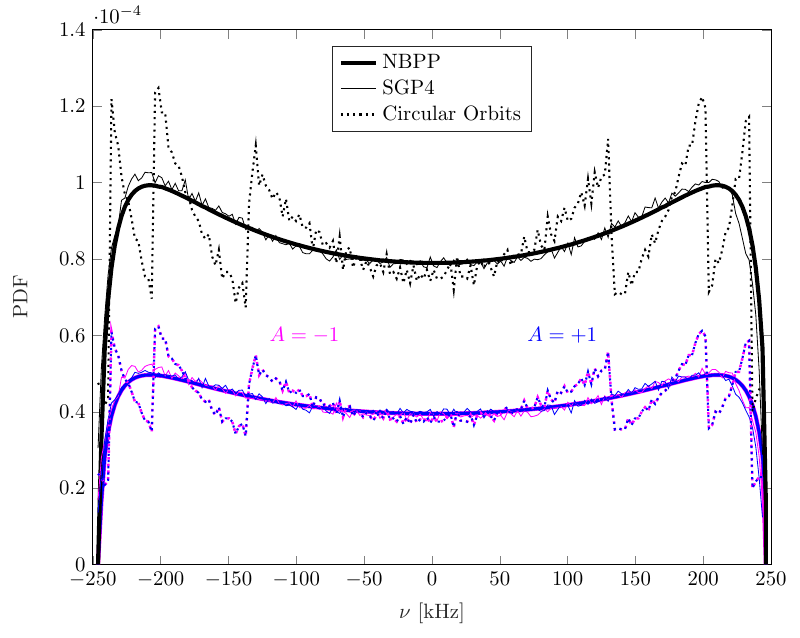}
        \caption{$\phi_{\rm{u}} = 90^\circ$ (equator), $\psi_{\min} = 30^{\circ}$}
    \end{subfigure}%
    ~ 
    \begin{subfigure}[b]{0.5\textwidth}
        \centering
        \includegraphics[width=1\textwidth]{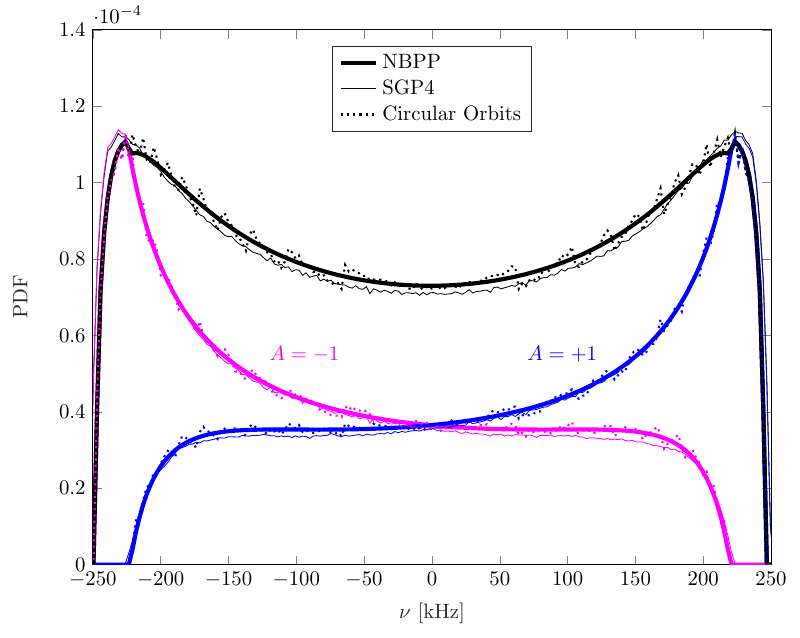}
        \caption{$\phi_{\rm{u}} = 30^\circ$ (edge of human population), $\psi_{\min} = 10^{\circ}$ }
    \end{subfigure}
    \caption{Doppler shift PDF of the NBPP compared with SGP4 and Circular Orbit simulations.}%NBPP, the circular orbits model with the system parameters in Table \ref{tab:sys_params}, and the SGP4 simulator.}
    \label{fig:DopplerPDF}
\end{figure*}

Consider two extreme examples: a user at the equator with limited visibility due to obstacles ($\phi_{\rm{u}} = 90^\circ$, $\psi_{\min} = 30^\circ$), and a user at the edge of the human population with wider visibility due to less obstacles ($\phi_{\rm{u}} = 30^\circ$, $\psi_{\min} = 10^\circ$). In Fig. \ref{fig:DopplerPDF}(a), the Doppler shift PDF of the NBPP for a user at the equator is symmetric and identical for ascending/descending satellites. In Fig. \ref{fig:DopplerPDF}(b), the Doppler shift PDF of the NBPP for a user at the edge of the human population maintains a similar shape and is symmetric, albeit with slightly more dense peaks. However, in this case the distribution of ascending/descending satellites is asymmetric, but they are identical distributions up to a sign change. Moreover, observe that the Doppler PDFs contain smooth peaks just before $\pm \nu_{\max}$, which differs from Jakes' Doppler spectrum { and the Doppler bound distribution in \cite{AlHourani2024} which have} sharp spikes directly at $\pm \nu_{\max}$. { Furthermore, the strengths of the Doppler bound distribution are that it has a closed-form PDF and its CDF fits well to SGP4 simulations for randomly placed users on Earth. On the other hand, its PDF has zero density at zero Doppler, which is not representative of the true Doppler PDFs in Fig. \ref{fig:DopplerPDF} that have a substantial non-zero density near zero Doppler.}
%In addition, we confirm that the direction approximation $\beta_a(\phi) \approx b$ is tight at the equator. 

Comparing the Doppler shift distribution of the NBPP compared with {circular orbits and SGP4 simulations}, we can determine the accuracy of the stochastic satellite model (similar observations apply to gain and delay). Recall that the satellites have rotational angles modelled by a {\em continuous} uniform random variable, which does not capture the discrete ascending nodes spaced by $s_{\rm{orb}}$. { Moreover, this restricts the number of unique orbit trajectories, each with their own statistics, that are visible to a user.} In Fig. \ref{fig:DopplerPDF}, this is observed as {\em bucketing artefacts} that are clusters of Doppler shifts. The NBPP model follows the general trend, but does not the capture the bucketing artefacts. These artefacts are reduced for mega-constellations with smaller orbital spacings. However, at higher latitudes the orbital spacing reduces; this is why the density of satellites is higher at higher latitudes. Hence, the NBPP model is more accurate for users at higher latitudes, as observed by comparing a user at latitude $0^\circ$ in Fig. \ref{fig:DopplerPDF}(a) with a user at latitude $60^\circ$ in Fig. \ref{fig:DopplerPDF}(b). { On the other hand, the additional perturbations in the SGP4 simulations strengthen the validity of the assumption of continuous uniform rotational angles with uncountably many unique orbits. Consequently, the Doppler shift PDF of the NBPP almost perfectly matches the realistic orbit simulations with real-world imperfections.}

\section{Mega-constellation Satellite Channel Model}
A mega-constellation composed of $N$ satellites is modelled according {to} the NBPP from the previous section with satellites at positions $\{(\Theta_i,\Phi_i)\}$ marked with directions $\{\beta_{A_i}(\Phi)\}$. The satellites induce independent gains $\{G_i\}$, delays $\{T_i\}$, and Doppler shifts $\{V_i\}$. We consider the scenario where no more than one visible satellite can transmit to the user, which occurs with probability $p_{\rm{a}}$.

The {\em central control unit (CCU)} of the mega-constellation randomly chooses the $i$-th visible satellite for transmission such that the indicator $I_{i}$ is equal to one and $I_{j}$ equals zero for all $j$ with $j\neq i$. The chosen satellite has spherical coordinates $(R, \tilde{\Theta}, \tilde{\Phi})$ with direction $\beta_{\tilde{A}}(\tilde{\Phi})$. As a function of these coordinates, it has gain $\tilde{G}$, propagation delay $\tilde{T}$, and Doppler shift $\tilde{V}$. With these signal properties, we propose the following {\em linear time-varying (LTV)} channel. 

\begin{definition}[Satellite channel model]\label{def:ch_model}
    The channel model for satellite mega-constellation communications to a user at polar angle $\phi_{\rm{u}}$ is an LTV channel with {\em spreading function}
\begin{align}
      \mathsf{S}(\tau, \nu) 
       &=  \sqrt{\tilde{G}} e^{-j2\pi \tilde{V}}  \delta(\tau - \tilde{T}) \delta(\nu - \tilde{V}) .
\end{align}
\end{definition}

This channel model makes {three simplifications} compared to the real system:

\subsubsection{CCU Independence}
There are many possible inter-satellite handover strategies for the CCU to balance the following desirable properties \cite{Chowdhury2006}:
\begin{itemize}
    \item {\em Maximum service time:} Choose the satellite that can serve the user the longest.
    \item {\em Maximum number of free channels:} Choose the satellite with the maximum number of free channels.
    \item {\em Minimum distance}\footnote{When a shadowing coefficient is side-information at the receiver, the satellite with maximum SNR is often 
 chosen.}: Choose the closest satellite so that the gain is maximised.
\end{itemize}

Rather than choosing one, the proposed channel model is assumed to use a CCU model that chooses a visible satellite independently of its position (e.g., any random visible satellite) from a cap of satellites with acceptable signal properties.

\subsubsection{Snapshot Independence}
The satellite channel model only models the real system at snapshot times $t_n = n\Delta t$ for $n \in \{0,1,2,\ldots\}$, and assumes independence between snapshots as in previous satellite models \cite{Okati2022}. This corresponds to a system that performs satellite handover on the scale of seconds, but in reality the same satellite may be used until it is out of range, on the scale of tens of seconds or minutes. Then, the {\em spreading function} of the real system varies over time as $\mathsf{S}(t;\tau,\nu)$ since the satellites in the constellation move along their orbits deterministically, and thus the snapshots with spreading functions $\mathsf{S}(t_n;\tau,\nu)$ are actually not independent in general. In fact, if we condition on the satellite position and direction at time $t_n$, then we deterministically know its position and direction at time $t_{n+1}$. However, over long time scales we actually do have snapshot independence. Therefore, we should interpret such models with snapshot independence as modelling the long-term statistics of the system, since they are based on the zero-order stationary distribution and not higher-order statistics. In the short-term, we might opt for a deterministic model \cite{Mcbain2025}.

\subsubsection{{No Fading}}{
There are many signal propagation models in the literature for satellite communications under various conditions, e.g., the shadowed-Rician fading model \cite{ShadowedRician2003}. These conditions often include troposphere effects caused by atmospheric gases, rain, fog, and clouds \cite{ITU2019}. Such external effects in the propagation medium, which are generally well understood, are ignored here in order to not obfuscate properties of the satellite channel due to the mega-constellation geometry alone, which is much less understood. However, they are readily combined with Definition \ref{def:ch_model} by additional fading coefficients and propagation paths; {see Appendix \ref{appendix:Rayleigh_fading} for an extension to the channel model with Rayleigh fading}.
}

\subsection{Scattering Function}
Since the satellite channel model is defined by a spreading function that is random, it is said to be a {\em random} LTV channel, where the channel realisation is randomly chosen according to the channel law. Such channels have a simplified description when the scatterers (the satellites) are uncorrelated and the channel taps are jointly wide-sense stationary, referred to as {\em wide-sense stationary uncorrelated scattering (WSSUS)}. With this property, the channel is completely described by its second-order statistics in the form of a two-dimensional {\em scattering function} \cite[Ch. 1]{Hlawatsch2011} that gives a concise description of how transmitted information symbols spread their power in the delay-Doppler domain.

\begin{theorem}\label{thm:scat_func}
    The satellite channel model is WSSUS with scattering function
    \begin{align}
        \mathsf{C}(\tau,\nu) &= \frac{p_{\rm{a}}}{2c^2 \tau^2} \bigg[f_{\overline{V},\overline{T}|A}(\nu,\tau|{+1}) + f_{\overline{V},\overline{T}|A}(\nu,\tau|{-1})\bigg].
    \end{align}
\end{theorem}
The proof is given in Appendix D.

\begin{figure*}
\centering
    \begin{subfigure}[b]{0.3\textwidth}
    \centering
    \includegraphics[width=5cm]{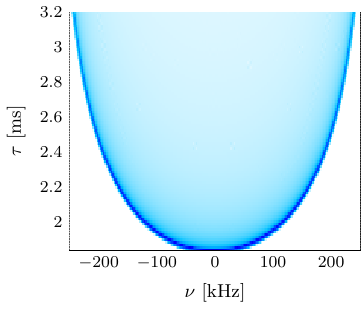}
    \caption{NBPP}
    \end{subfigure}
\quad
    \begin{subfigure}[b]{0.3\textwidth}
    \centering
    \includegraphics[width=5cm]{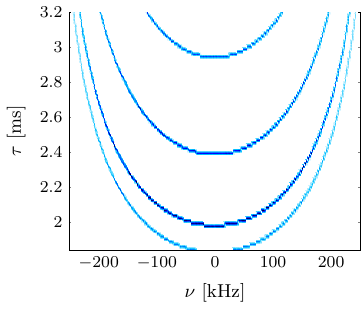}
    \caption{Circular Orbits}
    \end{subfigure}
\quad
    \begin{subfigure}[b]{0.3\textwidth}
    \centering
    \includegraphics[width=5cm]{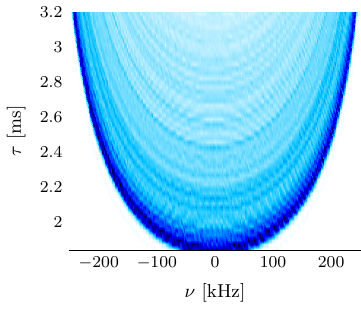}
    \caption{SGP4}
    \end{subfigure}

\caption{ Normalised scattering function of the satellite channel for $\phi_{\rm{u}} = 90^\circ$ (equator), $\psi_{\min} = 30^{\circ}$.}
 \label{fig:scat_func1}
\end{figure*}

\begin{figure*}
\centering
    \begin{subfigure}[b]{0.3\textwidth}
    \centering
    \includegraphics[width=5cm]{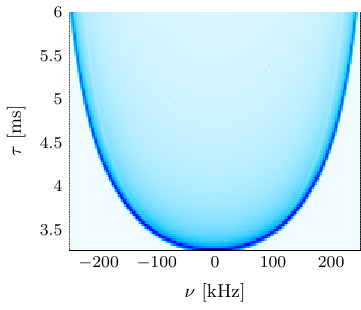}
    \caption{NBPP}
    \end{subfigure}
\quad
    \begin{subfigure}[b]{0.3\textwidth}
    \centering
    \includegraphics[width=5cm]{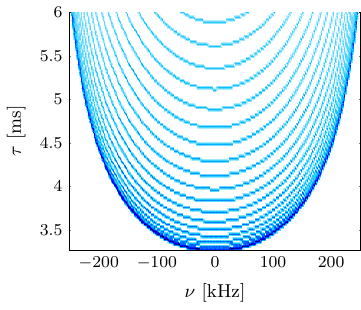}
    \caption{Circular Orbits}
    \end{subfigure}
\quad
    \begin{subfigure}[b]{0.3\textwidth}
    \centering
    \includegraphics[width=5cm]{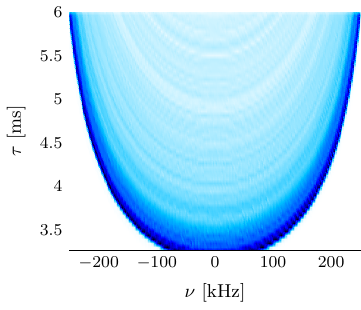}
    \caption{SGP4}
    \end{subfigure}

\caption{ Normalised scattering function of the satellite channel for $\phi_{\rm{u}} = 30^\circ$ (edge of human population), $\psi_{\min} = 10^{\circ}$. }
 \label{fig:scat_func2}
\end{figure*}

Observe that the scattering function in Theorem \ref{thm:scat_func} is proportional to the probability density of delay and Doppler, scaled by the instantaneous gain. As with the Doppler PDF, we approximate the delay-Doppler PDF using the finite difference method with resolutions $\Delta\nu = 2.61$ kHz and $\Delta\tau = 0.028$ ms ($\Delta\nu\Delta\tau = 0.073$) in the examples to follow.

% Explain the concentration of power near max doppler
The scattering function is given for a user at the equator in Fig. \ref{fig:scat_func1}, and for a user at the edge of the human population in Fig. \ref{fig:scat_func2}. In Fig. \ref{fig:scat_func1}, the delay spread is $\tau_{\max} - \tau_{\min} = 3.33 - 1.83 = 1.50$ ms and the Doppler spread is {$2\nu_{\max} = 2\times 246.2 = 492.4$} kHz. In Fig. \ref{fig:scat_func2}, the delay spread is $\tau_{\max} - \tau_{\min} = 6.10 - 3.30 = 2.80$ ms and the Doppler spread is {$2\nu_{\max} = 2\times 246.8 = 493.6$} kHz. The support of the delay-Doppler distribution is observed in each scenario as white patches of zero density, surrounding a convex hull of non-zero densities for pairs of delays and Doppler shifts. Observe that the power density is greatest near the edge of the support. This can be explained by the fact that the Doppler PDFs have high densities near their min. and max. Doppler shifts which occur on the boundary of the support. { Consequently, this causes a U-curve due to the strong correlation between delay and Doppler.}
In addition, the bucketing artefacts in Fig. \ref{fig:DopplerPDF} are once again observed in the scattering functions of the {circular orbits}, resulting in blue lines inside the convex hulls {that correspond to the unique orbit trajectories and thus} are reduced for smaller orbital spacings or higher user latitudes. { Observe that the SGP4 simulations do not have these bucketing artefacts due to the additional orbit perturbations, as in the Doppler distributions, and closely match the NBPP. This is strong evidence for the applicability of the stochastic channel model. The following analysis of global channel parameters will make these observations more concrete.}
\vspace{-4mm}

\subsection{Global Channel Parameters}

Oftentimes, we are only interested in global channel parameters of the channel. These parameters include the path loss (the average power gain of a random satellite) and the channel spread (the spread of the channel in terms propagation delay and Doppler shift). These are all properties of the scattering function.

\begin{definition}
    For the satellite channel model, the path loss is $P_L = - 10\log_{10}(\rho^2)$ with
    \begin{align}
        \rho^2 = p_{\rm{a}}\mathbb{E}[\overline{G} ]  ,
    \end{align}
    the RMS delay spread is
    \begin{align}
    \sigma_{\tau} = \rho^{-1} \sqrt{ p_{\rm{a}}\mathbb{E}[(\overline{T} - \overline{\tau})^2 \overline{G} ]  }  
    \end{align}
    with mean delay 
        \begin{align}
    \overline{\tau} = \frac{p_{\rm{a}}}{\rho^2}\mathbb{E}[\overline{T} \cdot\overline{G} ]    ,
    \end{align}
        and the RMS Doppler spread is
    \begin{align}
    \sigma_{\nu} =  \rho^{-1} \sqrt{p_{\rm{a}}\mathbb{E}[\overline{V}^2 \cdot \overline{G} ]}
    \end{align}
    with mean Doppler shift $\overline{\nu}=0$. These are readily computed using the scattering function from Theorem \ref{thm:scat_func}.
\end{definition}

\begin{table*}[]
    \centering
    \begin{tabular}{ccccccc}
\toprule
& \multicolumn{2}{c@{}}{NBPP} & \multicolumn{2}{c@{}}{Circular Orbits} & \multicolumn{2}{c@{}}{SGP4}\\
\cmidrule(l){2-3}
\cmidrule(l){4-5}
\cmidrule(l){6-7}
\multirow{2}{*}[0pt]{\textbf{Parameters}} & 
$\phi_{\rm{u}}=90^\circ$ & $\phi_{\rm{u}}=30^\circ$  &$\phi_{\rm{u}}=90^\circ$ & $\phi_{\rm{u}}=30^\circ$  &$\phi_{\rm{u}}=90^\circ$ & $\phi_{\rm{u}}=30^\circ$ \\
 & 
$\psi_{\min}=30^{\circ}$ & $\psi_{\min}=10^{\circ}$ &$\psi_{\min}=30^{\circ}$ & $\psi_{\min}=10^{\circ}$ &$\psi_{\min}=30^{\circ}$ & $\psi_{\min}=10^{\circ}$\\
\cmidrule(l){1-1}
\cmidrule(l){2-3}
\cmidrule(l){4-5}
\cmidrule(l){6-7}

Path loss, $P_L$ & $117.6$ dB & $122.6$ dB &$117.6$ dB &$122.6$ dB & $117.6$ dB & $122.5$ dB\\
Mean delay, $\overline{\tau}$ & $2.5$ ms & $4.5$ ms &$2.5$ ms &$4.4$ ms & $2.5$ ms & $4.4$ ms\\
RMS delay spread, $\sigma_{\tau}$ & $0.43$ ms & $0.80$ ms &$0.43$ ms &$0.80$ ms & $0.43$ ms & $0.82$ ms\\
Mean Doppler, $\overline{\nu}$ & $0$ kHz & $0$ kHz &$-0.12$ kHz&$-0.08$ kHz& $-0.86$ kHz& $0.16$ kHz \\
RMS Doppler spread, $\sigma_{\nu}$ & $134.5$ kHz & $137.9$ kHz &$134.2$ kHz &$138$ kHz & $137.9$ kHz & $140.5$ kHz\\

\bottomrule
\end{tabular}
    \caption{Global channel parameters of the satellite channel.}
    \label{tab:global_ch_params}
\end{table*}

These are the classical definitions that average the global channel parameters with respect to a measure of the probability density of a satellite position scaled by the gain at that position, explaining the presence of $\overline{G}$ in the above definitions. Another perspective is to normalise the scattering function as $\mathsf{C}(\tau,\nu)/\rho^2$, giving a probability measure whose first and second moments correspond to the given definitions.

    \begin{proposition}
        The path loss of the satellite channel model is $P_L = - 10\log_{10}(\rho^2)$ with
        \begin{align}
    \rho^2 &= \frac{p_{\rm{a}}}{p_{\rm{sat}}}\int_{g_{\min}}^{g_{\max}} p_{\rm{cap}}\left(\mathsf{G}^{-1}(g)\right) dg .
\end{align}
    \end{proposition}
    \begin{proof}
        Use the identity $\mathbb{E}[\overline{G}] = \int_{0}^{\infty} (1-F_{\overline{G}}(g)) dg$ and Theorem \ref{thm:cdfs}.
    \end{proof}

Consider the global channel parameters in Table \ref{tab:global_ch_params} for two extremes of user latitudes: a user at the equator and a user at the edge of the population. These results say that the mean delay is $2.5-4.5$ ms, highlighting the high latency in LEO satellite communications. In addition, the path loss is $118$ dB in the best case, and an additional $5$ dB in the worst case, requiring a significant transmission power to overcome.

Now let us consider the {\em channel spread} $2\sigma_{\tau} \sigma_{\nu}$. A channel is {\em underspread} if the channel spread is less than one, which is often the case in mobile communication channels. The underspread property describes channels that do not jointly spread far from the mean in terms of delay and Doppler shift; the channel may spread significantly in either delay or Doppler shift, but not both together. The examples in Table \ref{tab:global_ch_params} give a channel spread in the hundreds, indicating that the satellite channel {(without compensation)} is highly overspread. That is, if we draw a bounding box according to $\sigma_{\tau}$ and $\sigma_{\nu}$ in the delay-Doppler domain, the power leakage outside of this box is very high. This is due to the strong coupling between delay and Doppler shift since they increase together at a similar rate.

{Finally, observe in Table \ref{tab:global_ch_params} that the global channel parameters of the stochastic model are very close to the global channel parameters of the circular orbits and SGP4 simulations. This supports our earlier observations that the scattering functions are similar.}

\section{Conclusion}
In conclusion, a satellite channel model was proposed that models communications between a LEO mega-constellation, with rapidly-moving LEO satellites, and a stationary user on Earth. The channel model was derived using the stochastic geometry of satellites randomly distributed on a satellite sphere according to a NBPP, marked with random directions to model the trajectories of their orbits. A stochastic satellite mega-constellation channel model was proposed using the derived probability distributions of signal propagation properties including channel gain, propagation delay, and Doppler shift. 

Since the Doppler shift was based on the mega-constellation geometry, it improves upon the commonly used Jakes' Doppler spectrum in the scenario of LEO satellite communications. { The derived Doppler shift distributed revealed new properties that had not been studied in the literature, such as the dependence on ascending/descending and the statistics of individual orbit trajectories.} The generality of this characterisation makes it readily extended to a Doppler shift function $\mathsf{V}'(\theta,\phi) = \mathsf{V}(\theta,\phi) - \mathsf{V}_0(\theta,\phi)$ with pre-compensation $\mathsf{V}_0(\theta,\phi)$ (e.g., the Doppler shift at the centre of a beam). The statistics of mega-constellation communications channel were characterised by the scattering function of the stochastic channel model. {Since the Doppler shift distribution and the scattering function tightly matched realistic simulations using the SGP4 orbit simulator, the channel model is suitable for modelling the long-term behaviour of the channel. For modelling the short-term behaviour, the stochastic modelling approach could be extended to semi-stochastic modelling as in \cite{Mcbain2025}.} 

{ In addition, the scattering function and its properties demonstrated {\em how}} satellite communications is highly overspread in the delay-Doppler domain. Since wireless channels in the literature are typically underspread, this suggests the need for new approaches---such as delay-Doppler communications \cite{OTFS2022}---and { offers new insights to handle this unique} setting of highly mobile satellites.

\appendices

\section{Proof of Theorem \ref{thm:doppler}}

Consider the special case of a satellite moving in direction $\beta$ in the $\boldsymbol{\hat{\theta}}$-$\boldsymbol{\hat{\phi}}$ plane. The satellite velocity vector is 
\begin{align}
    \mathbf{v} &= v(\cos(\beta) \boldsymbol{\hat{\theta}} - \sin(\beta) \boldsymbol{\hat{\phi}})\\
    &= v\begin{bmatrix}
           -\cos(\beta)\sin(\theta) -\sin(\beta)\cos(\phi)\cos(\theta)\\
           \cos(\beta)\cos(\theta) -\sin(\beta)\cos(\phi)\sin(\theta) \\
           \sin(\beta)\sin(\phi)
         \end{bmatrix} ,
\end{align}
the line-of-sight vector from the { user to the satellite} is
\begin{align}
    \mathbf{d} &= \mathbf{s} - \mathbf{u}\\
    %&= R\hat{\mathbf{r}} - \mathbf{u}\\
    &= \begin{bmatrix}
           R\sin(\phi)\cos(\theta) \\
           R\sin(\phi)\sin(\theta) - r\sin(\phi_{\rm{u}})\\
           R\cos(\phi) - r\cos(\phi_{\rm{u}})
         \end{bmatrix} ,
\end{align}
and the user-satellite distance is $||\mathbf{d}||=\sqrt{r^2 + R^2 - 2rR \cos(\sigma(\theta,\phi) )}$. The Doppler shift function is the projection of $\mathbf{v}$ in the direction of $\mathbf{d}$ as $\mathbf{v} \cdot \hat{\mathbf{d}}$. 

However, the direction $\beta$ depends on the polar angle of the satellite on its orbit. Consider the spherical triangle of the satellite at $(R,\theta,\phi)$, the north pole at $(R,0,0)$, and the intersection point $(R,\pi/2,\overline{b})$ between the orbital plane (with ascending node at zero longitude) and the y-z plane. By the spherical sine rule, we have

\begin{align}
    \frac{\sin(\pi/2 - |\beta|)}{\sin(\pi/2 - b)} &= \frac{\sin(\pi/2)}{\sin(\phi)}\\
    \Rightarrow |\beta| &= \cos^{-1}\left(\frac{\cos(b)}{\sin(\phi)}\right)
\end{align}
which is the direction for ascending satellites, otherwise it is the negative.

\section{Proof of Theorem \ref{thm:p_cap}}
A satellite is visible by the user if it is in the cap surface $\mathsf{Cap}$ according to central angle $\sigma_1$. More generally, consider a cap surface according to central angle $\sigma$. For a fixed polar angle $\phi$, the visible surface becomes a visible curve at each polar angle $\phi$ for a {$\theta$ that is in an interval with endpoints on the edge of the cap and whose length is denoted by $L(\phi; \sigma)$}. Note that we only require an interval since $\mathsf{Cap}$ is a convex set with respect to $\theta$ and $\phi$. Therefore, the interval length is specified by finding the distance between the two roots of $\sigma(\theta,\phi) = \sigma_1$ with respect to $\theta$, and when the roots do not exist the length is $2\pi$ which corresponds to latitude lines that do not hit the cap boundary. This yields
\begin{align}
    &L(\phi; \sigma)\notag\\
    &= \begin{cases} 
      2\pi & 0 \leq \phi \leq \max\{0, \sigma_1 - \phi_{\rm{u}}\} \\
       L_1(\phi; \sigma) & \max\{0, \sigma_1 - \phi_{\rm{u}}\} < \phi < \phi_{\rm{u}} + \sigma_1 \\
      0 & \text{otherwise}\end{cases} 
\end{align}
as the piecewise arc length function for the non-sliced cap. Conditioned on polar angle $\phi$, the probability of the visible curve is $ L(\phi; \sigma)/2\pi$, since $\Theta$ is uniformly distributed.

Let $f_{\Theta, \Phi}(\theta, \phi) = f_{\Phi}(\phi)/2\pi$ be the density function of an NBPP on a spherical surface. The probability of observing a single satellite in a cap $\mathsf{Cap}(\sigma)$ of central angle $\sigma$ is the surface integral
\begin{align}
    p_{\rm{cap}}(\sigma) &= \int_{\mathsf{Cap}(\sigma)} f_{\Theta, \Phi}(\theta,\phi)  d\theta d\phi\\
    &= \frac{1}{2\pi}\int_{\mathsf{Cap}(\sigma)} f_{\Phi}(\phi)   d\theta d\phi\\
    %&=\frac{1}{2\pi} \int_{0}^{\pi} f_{\Phi}(\phi) \left[  \int_{\theta_0(\phi; \sigma)}^{\theta_1(\phi; \sigma)}  d\theta \right] d\phi\\
    &= \frac{1}{2\pi}\int_{0}^{\pi} f_{\Phi}(\phi) L(\phi; \sigma) d\phi
\end{align}
for $0 \leq \phi_{\rm{u}} \leq \phi/2$ since each polar angle with $\pi/2 < \phi_{\rm{u}} \leq \pi$ corresponds to the probability at polar angle $\pi - \phi$.

\section{Proof of Theorem \ref{thm:Doppler_CDF}}

A satellite in the user's visible cap must satisfy the inequality
$\sigma(\theta,\phi) \leq \sigma_1$, for $\overline{b} \leq \phi \leq \pi - \overline{b}$. The boundary of this cap
(the thick black curves in Fig. \ref{fig:integration_region}) is found by solving this with equality, which can be expressed in terms of $L_1(\phi; \sigma_1)$, the arc length between the two points on the boundary of the cap at polar angles $\phi$, due to symmetry about $\theta_{\rm{u}}$. This corresponds to the region of integration over the cap. The Doppler CDF requires integrating above a curve that corresponds to a constant Doppler shift (see the coloured contours in Fig. \ref{fig:integration_region}). Since this is intractable, we include an indicator function in the integrand to indirectly specify an additional constraint on the region of integration over the cap. This gives the CDF $F_{\overline{V}|A}(\nu|a) = \mathbb{E}[\mathbbm{1}{\{\mathsf{V}_a(\theta,\phi) \leq \nu\}}|\mathbbm{1}{\{(\Theta,\Phi) \in \mathsf{Cap}\}} = 1]$.

\begin{figure*}
    \centering
    \begin{subfigure}[b]{0.45\textwidth}
        \centering
        \includegraphics[width=0.95\textwidth]{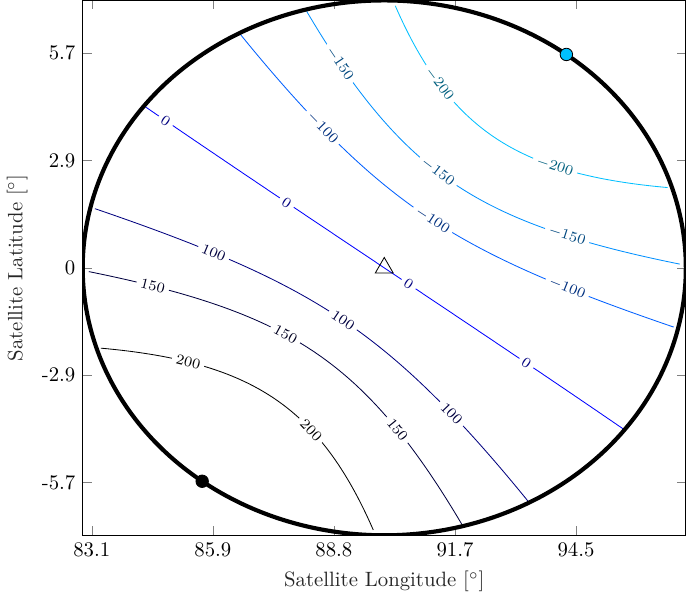}
        \caption{$\phi_{\rm{u}} = 90^{\circ}$, $\psi_{\min} = 30^{\circ}$, $A=+1$}
    \end{subfigure}%
    ~ 
    \begin{subfigure}[b]{0.45\textwidth}
        \centering
        \includegraphics[width=0.95\textwidth]{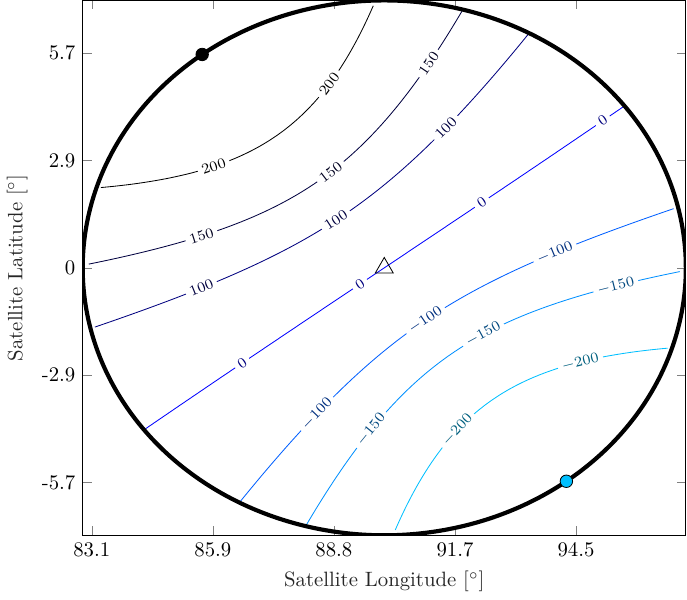}
        \caption{$\phi_{\rm{u}} = 90^{\circ}$, $\psi_{\min} = 30^{\circ}$, $A=-1$}
        %\caption{$\phi_{\rm{u}} = 30^{\circ}$, $\psi_{\min} = 10^{\circ}$ (edge of human population)}
    \end{subfigure}

        \bigskip
    \begin{subfigure}[b]{0.45\textwidth}
        \centering
        \includegraphics[width=0.95\textwidth]{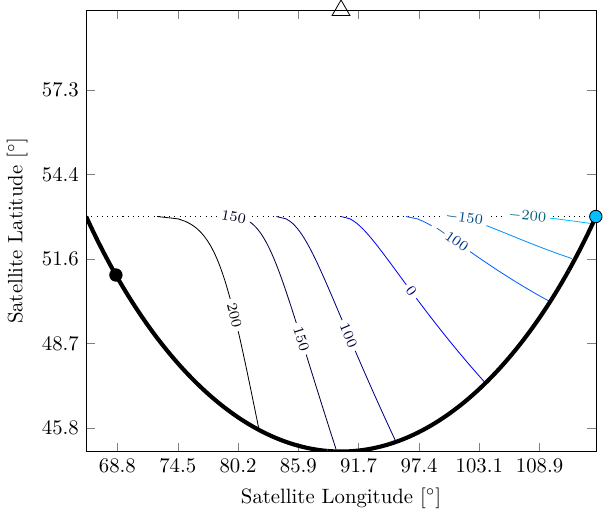}
        \caption{$\phi_{\rm{u}} = 30^{\circ}$, $\psi_{\min} = 10^{\circ}$, $A=+1$}
        %\caption{$\phi_{\rm{u}} = 30^{\circ}$, $\psi_{\min} = 10^{\circ}$ (edge of human population)}
    \end{subfigure}
        ~ 
    \begin{subfigure}[b]{0.45\textwidth}
        \centering
        \includegraphics[width=0.95\textwidth]{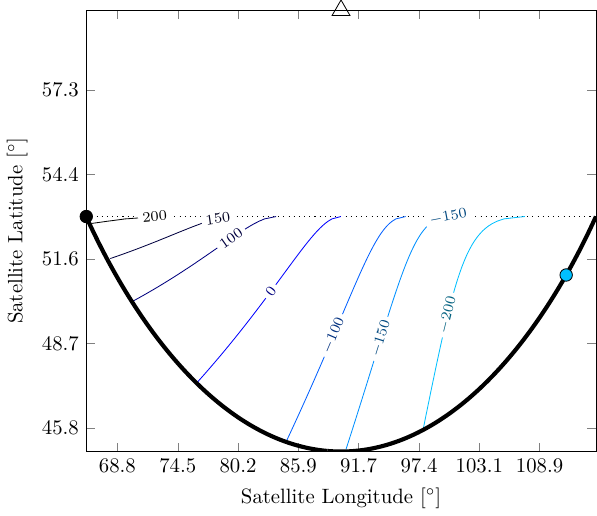}
        \caption{$\phi_{\rm{u}} = 30^{\circ}$, $\psi_{\min} = 10^{\circ}$, $A=-1$}
        %\caption{$\phi_{\rm{u}} = 30^{\circ}$, $\psi_{\min} = 10^{\circ}$ (edge of human population)}
    \end{subfigure}
    \caption{{The integration region $\mathsf{Cap}$ with contours of the Doppler shift function (in kHz). The thick black curve is the boundary of $\mathsf{Cap}$. The min. Doppler shift is at the blue circle marker and the max. Doppler shift is at the black circle marker. The user is located at the triangle marker.}}
    \label{fig:integration_region}
\end{figure*}

\section{Proof of Theorem \ref{thm:scat_func}}

The scattering function of a WSSUS channel satisfies
\begin{align}
    &\mathbb{E}[\mathsf{S}(\tau,\nu) \mathsf{S}^*(\tau',\nu')]\notag\\
    &= \mathbb{E}[\mathsf{G}(\mathsf{T}^{-1}(\tau)) \delta(\tau - \tilde{T}) \delta(\nu - \tilde{V}) \delta(\tau' - \tilde{T}) \delta(\nu' - \tilde{V})]\\
    &= \mathbb{E}[\mathsf{G}(\mathsf{T}^{-1}(\tau)) \delta(\tau - \tilde{T})\delta(\nu - \tilde{V})] \delta(\tau - \tau') \delta(\nu - \nu') \\
    &= \mathsf{C}(\tau,\nu) \delta(\tau - \tau') \delta(\nu - \nu')  .
\end{align}

Observe that
\begin{align}
    &\mathsf{C}_a(\tau,\nu) \notag\\
    &:= \mathsf{G}(\mathsf{T}^{-1}(\tau)) \mathbb{E}[\delta(\tau - \tilde{T})\delta(\nu - \tilde{V})]\\
    &= p_{\rm{a}}\mathsf{G}(\mathsf{T}^{-1}(\tau))\notag\\
    &\mathbb{E}[ \delta(\tau - \mathsf{T}(\sigma(\Theta,\Phi)))\delta(\nu - \mathsf{V}_a(\Theta,\Phi))| \mathbbm{1}{\{(\Theta,\Phi)\in\mathsf{Cap}\}} = 1]\\
    &= p_{\rm{a}} \mathsf{G}(\mathsf{T}^{-1}(\tau)) \frac{\partial^2}{\partial \nu \partial \tau} F_{\overline{V},\overline{T}|A}(\nu,\tau|a)\\
   &= p_{\rm{a}} \mathsf{G}(\mathsf{T}^{-1}(\tau)) f_{\overline{V},\overline{T}|A}(\nu,\tau|a) 
\end{align}
where we differentiated inside the integral (the expectation) and used the identity $\frac{d}{dx} \mathbbm{1}\{x \leq y\} = \delta(x - y)$ (where the indicator function is treated as a distribution) to substitute the joint PDF from Corollary \ref{thm:DD_CDF}. Then $\mathsf{C}(\tau,\nu) = 0.5\mathsf{C}_{+1}(\tau,\nu) + 0.5\mathsf{C}_{-1}(\tau,\nu)$.

\section{Extended Version of the Stochastic Channel Model with Rayleigh Fading}\label{appendix:Rayleigh_fading}

As a concrete example, let us consider how to include Rayleigh fading into the stochastic channel model. Consider a Rayleigh random variable $X$ with scale parameter $\sigma^2_{\rm Rayleigh}=1/2$. The power is $Z=X^2$ with exponential PDF $f_{Z}(z) = e^{-z}$. Including the free-space path loss, we have the updated channel gain power as
\begin{align}
    Y = Z \tilde{G} .
\end{align}

Now let us consider how to update the probability distribution and the scattering function:

1) The channel gain CDF for a visible satellite is now
\begin{align}
    F_{Y}(y) &=\mathbb{P}(Y \leq y)\notag\\
    &=\mathbb{P}(\overline{G} \leq y/Z)\notag\\
    &= \mathbb{E}[F_{\overline{G}}(y/Z)] \label{eq:gen_gain_CDF}\\
    &= \int_{0}^{\infty} e^{-z} F_{\overline{G}}(y/z) dz \notag\\
    &= \int_{y/g_{\max}}^{\infty} e^{-z} \left[1 - \frac{p_{\rm cap}(\mathsf{G}^{-1}(y/z))}{p_{\rm sat}}\right] dz \notag\\
    &= e^{-y/g_{\min}} - \frac{1}{p_{\rm sat}}\int_{y/g_{\max}}^{y/g_{\min}} e^{-z} p_{\rm cap}(\mathsf{G}^{-1}(y/z)) dz
\end{align}
where we used our Theorem 4 to substitute $F_{\overline{G}}(y/z) = 1 - p_{\rm cap}(\mathsf{G}^{-1}(y/z))/p_{\rm sat}$. This CDF is shown in Fig. \ref{fig:CDF_Rayleigh_gain}.

2) The scattering function is now denoted by $\mathsf{C}'(\tau,\nu)$, which relates to the original scattering function $\mathsf{C}(\tau,\nu)$ without fading as
\begin{align}
    \mathsf{C}'(\tau,\nu) &= \mathbb{E}[Z \mathsf{G}(\mathsf{T}^{-1}(\tau)) \delta(\tau - \tilde{T})\delta(\nu - \tilde{V})]  \notag\\
    &= \mathbb{E}[Z \delta(\tau - \tilde{T})\delta(\nu - \tilde{V})] \mathsf{C}(\tau,\nu)\label{eq:gen_scat_func}\\
    &= \mathsf{C}(\tau,\nu) \notag
\end{align}
where we substituted $\mathbb{E}[Z \delta(\tau - \tilde{T})\delta(\nu - \tilde{V})]=\mathbb{E}[Z]=1$. That is, there is no change to the scattering function since the average power of the fading coefficients is $1$ and they are independent of delay $\tilde{T}$ and Doppler $\tilde{V}$. 

\begin{remark}For a shadowing model with general parameters dependent on satellite position, the expectations in (\ref{eq:gen_gain_CDF}) and (\ref{eq:gen_scat_func}) for the channel gain CDF and scattering function, respectively, would not necessarily simplify as in the no-fading case where channel gain was solely dependent on the central angle $\sigma=\sigma(\tilde{\Theta},\tilde{\Phi})$. However, if the fading parameters only depend on central angle (or elevation angle), these expectations can still be computed in terms of $p_{\rm cap}(\sigma)$ almost identically to the Rayleigh fading case. Otherwise, these expectations would have to be evaluated similarly to the Doppler shift distribution in Theorem 5 without any simplifications, recalling that the Doppler shift function did not solely depend on the central angle.

\begin{figure}
    \centering
    \includegraphics[width=1\linewidth]{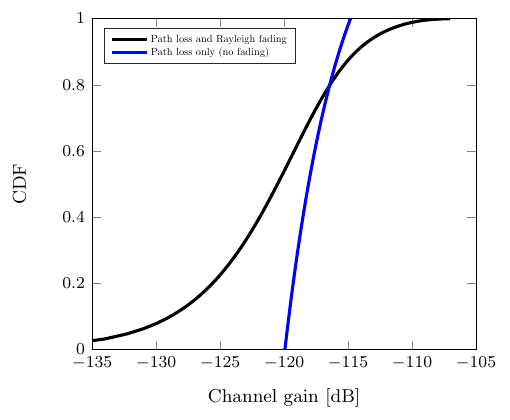}
    \caption{CDF of channel gain with and without Rayleigh fading with scaling parameter $\sigma^2_{\rm Rayleigh}=1/2$ for the NBPP with respect to a user at $\phi_{\rm u}=90^\circ$ with $\psi_{\rm min}=30^\circ$.}
    \label{fig:CDF_Rayleigh_gain}
\end{figure}

%the fading coefficients would be conditionally independent to delay $\tilde{T}=\tilde{T}(\Theta,\Phi)$ and Doppler $\tilde{V}=\tilde{V}(\Theta,\Phi)$ with conditioning on the satellite position $(\Theta,\Phi)$. Therefore, in this case we would have to evaluate $\mathbb{E}[X^2  \delta(\tau - \tilde{T})\delta(\nu - \tilde{V})]$, which would only simplify as in the no-fading case if $\mathbb{E}[X^2]$ depends only on the satellite position through the central angle so that we could similarly substitute $\sigma=\mathsf{T}^{-1}(\tau)$.
\end{remark}

%\section*{Acknowledgment}

\ifCLASSOPTIONcaptionsoff
  \newpage
\fi

\bibliography{refs}
\bibliographystyle{IEEEtran}

\end{document}